\documentclass[superscriptaddress,aps,pra,twocolumn,nofootinbib]{revtex4-1}

\usepackage{amsmath,amssymb,amsthm}
\usepackage{easybmat}
\usepackage{hyperref}
\usepackage[pdftex]{graphicx}
\usepackage{times,txfonts}
\usepackage{braket}
\usepackage{color}
\usepackage{natbib}

\newcommand{\be}{\begin{equation}}
\newcommand{\ee}{\end{equation}}
\newcommand{\ba}{\begin{eqnarray}}
\newcommand{\ea}{\end{eqnarray}}
\newcommand{\ketbra}[2]{|#1\rangle \langle #2|}
\newcommand{\tr}{\operatorname{Tr}}

\newtheorem{thm}{Theorem}

\newtheorem{observation}{Observation}
\newtheorem{definition}{Definition}
\newtheorem{proposition}{Proposition}

\newcommand{\plane}[2]{$#1#2$\nobreakdash-plane}
\begin{document}
\title{On total correlations in a bipartite quantum joint probability distribution}   
   \author{C. Jebarathinam}
\email{jebarathinam@iisermohali.ac.in; jebarathinam@gmail.com}
\affiliation{Department of Physical Sciences
\\Indian Institute of Science Education and Research (IISER) Mohali, \\
Sector-81, S.A.S. Nagar, Manauli 140306, India.}

\date{\today}

\begin{abstract}
We discuss the problem of separating the total correlations in a given quantum joint probability distribution into nonlocality, contextuality and classical correlations.
Bell discord and Mermin discord which qunatify nonlocality and contextuality of quantum correlations
are interpreted as distance measures in the nonsignaling polytope. 
A measure of total correlations is introduced to divide the total amount of correlations into a purely nonclassical part and a 
classical part. We show that quantum correlations satisfy additivity relations among these three measures.
\end{abstract}

\maketitle


\section{Introduction}
When measurements on an ensemble of entangled particles give rise to the violations of a Bell inequality \cite{bell64,BNL}, 
one may ask the question of EPR2 \cite{EPR2} whether all the particle pairs in the ensemble behave nonlocally or only some pairs are nonlocally correlated and 
the other pairs are locally correlated. 
EPR2 approach to quantum correlation consists 
in decomposing the given quantum joint probability distribution into nonlocal and local distributions 
to find out whether the correlation is 
fully nonlocal or it has local content. 
EPR2 showed that if the particles pairs are in the singlet state, they all behave nonlocally. However, 
EPR2 showed that nonmaximally entangled states cannot have nonlocality purely. Thus, total correlations arising from measurements on composite quantum systems
can be divided into a purely nonlocal part and a local part.

In Ref. \cite{Jeba}, Bell discord and Mermin discord has been proposed as measures of quantum correlations to quantify nonlocality and contextuality of 
symmetric discordant states \cite{WZQD,GTC,revDis}
and it has been shown that any quantum correlation can be decomposed as a convex mixture of a irreducible nonlocal correlation,
a irreducible contextual correlation and a local box which has null Bell discord and null Mermin discord.
The canonical decomposition for quantum correlations suggests that when measurements on an ensemble
of bipartite quantum system gives rise to Bell discord and Mermin discord
simultaneously, the ensemble can be divided into a purely nonlocal part, a contextual part and a local part which might have classical correlations. 

In this work, we discuss the analogous problem of dividing the total correlations in a given quantum state into a purely nonclassical part and a classical part 
\cite{HV,GPW,Modietal} in quantum joint probability distributions. 
We show that Bell discord and Mermin discord are interpreted as distance measures in the nonsignaling polytope  and thus they are analogous to geometric measure of
quantum discord \cite{Dakicetal}.
Inspired by this interpretation, we define a third distance measure to quantify the amount  of total correlations in quantum joint probability distributions.
We study additivity relation for quantum correlations in two-qubit systems.

The paper is organized as follows. In Sec. \ref{NS}, we review the bipartite nonsignaling boxes which have two inputs and two outputs per party. 
In Sec. \ref{measures}, we interpret
Bell discord and Mermin discord as distance measures and we define the third distance measure. 
In Sec. \ref{QC}, we investigate total correlations in quantum boxes using these three measures. In Sec. \ref{disconc}, we present conclusions.
\section{NS polytope of Bell-CHSH scenario}\label{NS}
Bell-CHSH scenario \cite{chsh} can be abstractly described in terms black boxes shared between two spatially separated observers; Alice and Bob input two variables $A_i$ and 
$B_j$ into the box and obtain two distinct outputs $a_m$ and $b_n$ on their part of the box ($i,j,m,n\in\{0,1\}$). 
The behavior of a given box is described by the set of $16$ joint 
probability distributions,
\ba
P(a_m,b_n|A_i,B_j)&=&\frac{1}{4}[1+(-1)^m\braket{A_i}+(-1)^n\braket{B_j}\nonumber\\
&&+(-1)^{m\oplus n}\braket{A_iB_j}],
\ea
where $\braket{A_iB_j}=\sum_{m=n}P(a_m,b_n|A_i,B_j)-\sum_{m\ne n}P(a_m,b_n|A_i,B_j)$ are joint expectation values,
and, $\braket{A_i}=P(a_0|A_i)-P(a_1|A_i)$ and $\braket{B_j}=P(b_0|B_j)-P(b_1|B_j)$ are marginal expectation values.
Here $\oplus$ denotes addition modulus $2$.
The set of nonsignaling boxes ($\mathcal{N}$) corresponding to this scenario forms an $8$ dimensional convex polytope which has $24$ extremal boxes \cite{Barrett}:  
they are $8$ PR-boxes,
\be
P^{\alpha\beta\gamma}_{PR}(a_m,b_n|A_i,B_j)=\left\{
\begin{array}{lr}
\frac{1}{2}, & m\oplus n=ij \oplus \alpha i\oplus \beta j \oplus \gamma\\ 
0 , & \text{otherwise},\\
\end{array}
\right. \label{NLV}
\ee
and $16$ deterministic boxes:
\be
P^{\alpha\beta\gamma\epsilon}_D=\left\{
\begin{array}{lr}
1, & m=\alpha i\oplus \beta,\\
   & n=\gamma j\oplus \epsilon \\
0 , & \text{otherwise}.\\
\end{array}
\right.
\ee

The NS polytope can be divided into two parts: nonlocal region and Bell polytope. The set of local boxes forms the 
Bell polytope ($\mathcal{L}$) which is a convex hull of the $16$ deterministic boxes: if $P(a_m, b_n|A_i,B_j)\in \mathcal{L}$,   
\ba
P(a_m, b_n|A_i,B_j)=\sum^{15}_{l=0}q_lP^l_{D}; \sum_lq_l=1. \label{LD}
\ea
Here $l=\alpha\beta\gamma\epsilon$. Quantum correlations obtained by measurements on composite systems form a convex subset of nonsignaling correlations
and is sandwitched between the nonlocal region and the Bell polytope \cite{BNL}.

Local reversible operations (LRO) are analogous to local unitary operations in quantum theory. 
It is known that Alice and Bob can not decrease entanglement and can not create entanglement from separability by local unitary operations
on the quantum states \cite{Hetal}, similarly, nonlocality and locality are invariant under LRO 
as they does not convert a deterministic box into a PR-box and vice versa. 
LRO include relabeling of party's inputs and outputs as follows: 
Alice changing her input $i\rightarrow i\oplus 1$, and changing her output conditioned on the input: $m\rightarrow m\oplus\alpha i\oplus\beta$. 
Bob can perform similar operations. 
\section{The three distance measures}\label{measures}
The distance measures are useful tool in quantum information theory to quantify nonclassicality quantum states and to divide the total correlations 
in a given quantum state into a nonclassical part and a purely a classical part \cite{Hetal,revDis,Modietal}. 
In Ref. \cite{Modietal}, measures of quantum correlations that go beyond entanglement were defined using the idea of distance measures
and it was shown that the distance of a given state to its closest product state
gives total correlations. 
We will define a distance measure which is nonzero iff the given JPD is nonproduct to quantify 
total correlations in quantum JPD.
\subsection{Bell discord}
All the Bell-CHSH inequalities \cite{ACHSH},
\ba
\mathcal{B}_{\alpha\beta\gamma} &:= &(-1)^\gamma\braket{A_0B_0}+(-1)^{\beta \oplus \gamma}\braket{A_0B_1}\nonumber\\
&+&(-1)^{\alpha \oplus \gamma}\braket{A_1B_0}+(-1)^{\alpha \oplus \beta \oplus \gamma \oplus 1} \braket{A_1B_1}\le2, \label{BCHSH}
\ea   
form eight facets for the Bell polytope. 
We may consider the eight Bell functions, $\mathcal{B}_{\alpha\beta\gamma}$, 
to form the eight orthogonal coordinates for the metric space in which 
distance is measured by $\mathcal{B}_{\alpha\beta}:=|\mathcal{B}_{\alpha\beta\gamma}|$ as they satisfy triangle inequality.
\begin{observation}
Bell functions satisfy the triangle inequality,
\begin{equation}
\mathcal{B}_{\alpha\beta}(P_1,P_2)\le \mathcal{B}_{\alpha\beta}(P_1)+\mathcal{B}_{\alpha\beta}(P_2). \label{Tr}
\end{equation} 
\end{observation}
\begin{proof}
Consider the following correlation,
\be
P=pP^{000}_{PR}+qP^{001}_{PR},
\ee
which has $\mathcal{B}_{00}(P)=4|p-q|$. Here $\mathcal{B}_{00}(P)$ can be regarded as measuring the distance between the boxes $P_1=pP^{000}_{PR}+(1-p)P_N$ and 
$P_2=qP^{001}_{PR}+(1-q)P_N$ which have $\mathcal{B}_{00}(P_1)=4p$ and $\mathcal{B}_{00}(P_2)=4q$. The triangle inequality in Eq. (\ref{Tr})
follows since $\mathcal{B}_{00}(P_1, P_2)=4|p-q|\le \mathcal{B}_{\alpha\beta}(P_1)+\mathcal{B}_{\alpha\beta}(P_2)= 4p+4q=4$. 
\end{proof}
The isotropic PR-boxes,
\be
P_{iPR}^{\alpha\beta\gamma}=p_{nl}P^{\alpha\beta\gamma}_{PR}+(1-p_{nl})P_N, \label{isoPR}
\ee
define the eight orthogonal coordinates in which each coordinate is a line joining a PR-box and white noise.
Geometrically for a given correlations, each $\mathcal{B}_{\alpha\beta}$ measures the distance of a box which is, in general, different than the given correlation. 
The white noise, $P_N$, which has 
$\mathcal{B}_{\alpha\beta\gamma}=0$ is at the origin. 
Since a PR-box can lie on top of only one facet, the distance of a PR-box from the origin is measured by only one of the Bell functions. 
For instance, the PR-box, $P^{00\gamma}_{PR}$, gives $\mathcal{B}_{00}=4$ and the rest of the $\mathcal{B}_{\alpha\beta}=0$; it is at the largest 
distance from the origin. Since the isotropic PR-boxes in Eq. (\ref{isoPR})  
lie along only one of the coordinates, they have $\mathcal{B}_{\alpha\beta}=4p_{nl}$ and the rest of the three Bell functions take zero. 
The deterministic boxes, which are at the faces of the Bell polytope, are simultaneously measured by
all the four Bell functions i.e., they give $\mathcal{B}_{\alpha\beta}=2$ for all $\alpha\beta$ as they lie on the hyperplane. 

Bell discord is constructed using the Bell functions as follows,
\be
\mathcal{G}=\min_i\mathcal{G}_i,
\ee
where $\mathcal{G}_1=\Big||\mathcal{B}_{00}-\mathcal{B}_{01}|-|\mathcal{B}_{10}-\mathcal{B}_{11}|\Big|$ and $\mathcal{G}_2$ and $\mathcal{G}_3$ are obtained by permuting 
$\mathcal{B}_{\alpha\beta}$ in $\mathcal{G}_1$. Here $0\le\mathcal{G}\le4$. The deterministic boxes have $\mathcal{G}=0$, whereas the PR-boxes have $\mathcal{G}=4$. 
As Bell discord is made up of $\mathcal{B}_{\alpha\beta}$, it also satisfies the triangle inequality. 
\begin{proposition}
If a nonextremal correlation has irreducible PR-box component, $\mathcal{G}$ measures how far the given correlation from a local box that does not have minimal
single PR-box excess in the metric space defined by the Bell functions.
\end{proposition}
\begin{proof}
Any NS correlation can be written as a convex combination of a irreducible PR-box and a local box which has $\mathcal{G}=0$ \cite{Jeba},
\be
P=\mathcal{G'}P^{\alpha\beta\gamma}_{PR}+(1-\mathcal{G'})P^{\mathcal{G}=0}_L. \label{gcano}
\ee
This decomposition implies that the correlation which has irreducible PR-box component lies between the line joining the single PR-box and the local box that 
does not have minimal single PR-box excess. Thus, Bell discord of the correlation in Eq. (\ref{gcano}) given by $\mathcal{G}(P)=4\mathcal{G}'$ gives  
distance of the given correlation from the $\mathcal{G}=0$ box in the canonical decomposition.
\end{proof}

Consider the following correlations,
\be
P=pP^{000}_{PR}+qP^{0000}_D.
\ee
For these correlations, $\mathcal{B}_{000}=p\mathcal{B}_{000}(P^{000}_{PR})+q\mathcal{B}_{000}(P_D)=4p+2q=2(p+1)$ and $\mathcal{G}=4p$. 
Notice that, $\mathcal{B}_{00}\ge \mathcal{G}$; $\mathcal{B}_{00}$ measures the distance of the correlation from the origin and is equal to the sum of the distance of 
the noisy deterministic box, $qP_D+(1-q)P_N$, and the noisy PR-box, $pP_{PR}+(1-p)P_N$, whereas $\mathcal{G}$ measures the distance of the correlation from the 
deterministic box and is equal to the distance of the correlation from the origin minus the distance of the noisy deterministic box. 

\textit{Violation of Bell-CHSH inequality versus Bell discord:-}
The canonical decomposition in Eq. (\ref{gcano}) gives $\mathcal{B}_{\alpha\beta\gamma}(P)=4\mathcal{G'}+l(1-\mathcal{G'})$, where $l=\mathcal{B}_{\alpha\beta\gamma}\left(P^{\mathcal{G}=0}_L\right)$.
Consider the case when $l\ge0$. If $\mathcal{G'}>\frac{1}{2}$, it is for sure that the correlation gives
the violation of the Bell-CHSH inequality. Now consider the following two cases.

(i) Suppose $\mathcal{B}_{\alpha\beta\gamma}\left(P^{\mathcal{G}=0}_L\right)=0$, the correlations can not give rise to the violation of the Bell-CHSH inequality when $0\le p\le\frac{1}{2}$.
Therefore, for the violation of the Bell-CHSH inequality upon increasing the PR-box content, first the box has to be lifted to the face of the Bell polytope by the PR-box content which happens at $\mathcal{G'}=\frac{1}{2}$. 

(ii) Suppose $\mathcal{B}_{\alpha\beta\gamma}\left(P^{\mathcal{G}=0}_L\right)=2$. Then any small amount of the PR-box content will give rise to the violation of the Bell-CHSH
inequality because the box lies on the face of the Bell polytope when $\mathcal{G'}=0$. 

Thus, the violation of Bell inequality depends on the amount of 
irreducible PR-box content as well as the local box in the 
cannical decomposition, whereas nonzero Bell discord depends only on the amount of irreducible PR-box content. 
Popescu and Rohrlich showed that all pure entangled states violate a Bell-CHSH inequality \cite{PRQB}. However, there are mixed entangled states which do not violate 
a Bell-CHSH inequality \cite{Hetal}.
The reason for the nonviolation of any Bell inequality by some entangled states is that the local box in the canonical decomposition does not have sufficient amount of magnitude for
the Bell operator to lift the correlation to go outside the Bell polytope. 
\subsection{Mermin discord}
We may as well consider the eight Mermin functions, 
\ba
\mathcal{M}_{\alpha\beta\gamma}&:=
&(\alpha\oplus\beta\oplus1)\{(-1)^{\beta}\braket{A_0B_1}\!+\!(-1)^{\alpha}\braket{A_1B_0}\}\nonumber\\ &&+(\alpha\oplus\beta)\{(-1)^{\gamma}\braket{A_0B_0}+(-1)^{\alpha\oplus\beta\oplus\gamma\oplus 1}\braket{A_1B_1}\}\nonumber\\
   && \text{for} \quad \alpha\beta\gamma=00\gamma,01\gamma;\nonumber \\
\mathcal{M}_{\alpha\beta\gamma}&:=&(\alpha\oplus\beta)\{(-1)^{\beta}\braket{A_0B_1}\!+\!(-1)^{\alpha}\braket{A_1B_0}\}\nonumber\\ &&+(\alpha\oplus\beta\oplus1)\{(-1)^{\gamma}\braket{A_0B_0}+(-1)^{\alpha\oplus\beta\oplus\gamma\oplus1}\braket{A_1B_1}\}\nonumber\\
   && \text{for} \quad \alpha\beta\gamma=10\gamma,11\gamma,
\ea 
to form eight orthogonal coordinates for the metric space in which $\mathcal{M}_{\alpha\beta}:=|\mathcal{M}_{\alpha\beta\gamma}|$ serve as the distance function.
The eight Mermin boxes, $P^{\alpha\beta\gamma}_M$: for $\alpha\beta\gamma=00\gamma,10\gamma$, 
\be
P_M^{\alpha\beta\gamma}(a_m,b_n|A_i,B_j)=\left\{
\begin{array}{lr}
\frac{1}{4}, & i\oplus j =0 \\
\frac{1}{2}, & m\oplus n=i\cdot j \oplus  \alpha  i \oplus  \beta j \oplus \gamma\\ 
0 , & \text{otherwise},\nonumber\\
\end{array}
\right.
\ee
and, for $\alpha\beta\gamma=01\gamma,11\gamma$,
\be
P_{M}^{\alpha\beta\gamma}(a_m,b_n|A_i,B_j)=\left\{
\begin{array}{lr}
\frac{1}{4}, & i\oplus j =1 \\
\frac{1}{2}, & m\oplus n=i\cdot j \oplus  \alpha  i\oplus \beta j \oplus  \gamma  \\ 
0 , & \text{otherwise}\\
\end{array} \label{Mmmm}
\right.
\ee
lie along extremum of only one of the coordinates. Therefore, the distance of the isotropic Mermin boxes,
\be
P_{iM}^{\alpha\beta\gamma}=p_{c}P^{\alpha\beta\gamma}_M+(1-p_{c})P_N, \label{isoM}
\ee
are measured by only one of the Mermin functions i.e., they have $\mathcal{M}_{\alpha\beta}=2p_{c}$ and the rest of the three Mermin functions take zero.

Mermin discord is constructed using the Mermin functions as follows,
\be
\mathcal{Q}=\min_i\mathcal{Q}_i.
\ee
Here $\mathcal{Q}_1=\Big||\mathcal{M}_{00}-\mathcal{M}_{01}|-|\mathcal{M}_{10}-\mathcal{M}_{11}|\Big|$ and $\mathcal{Q}_2$ and $\mathcal{Q}_3$ are obtained by permuting 
$\mathcal{M}_{\alpha\beta}$ in $\mathcal{Q}_1$. Since the distance of the PR-boxes and the deterministic boxes are simultaneously measured by two Mermin functions (i.e., they
lie on the hyperplane), they have $\mathcal{Q}=0$. The isotropic Mermin boxes in Eq. (\ref{isoM}) have $\mathcal{Q}=2p_c$.
\begin{proposition}
If a given correlation has nonzero Mermin discord, $\mathcal{Q}$ measures the distance of the given correlation from a correlation that does not have minimal single
Mermin box excess in the metric space of Mermin functions.
\end{proposition}
\begin{proof}
Any NS correlation can be decomposed as a convex mixture of a $\mathcal{Q}=2$ box which lies on extremum of one of the coordinates and a $\mathcal{Q}=0$ box \cite{Jeba},
\be
P=\mathcal{Q'}P^{\alpha\beta\gamma}_{\mathcal{Q}=2}+(1-\mathcal{Q'})P_{\mathcal{Q}=0}. \label{canoQ}
\ee
This decomposition implies that the correlation that has irreducible Mermin box component lies between a line joining the $\mathcal{Q}=2$ box and the $\mathcal{Q}=0$ box.
Thus, Mermin discord of the correlation in Eq. (\ref{canoQ}) given by $\mathcal{Q}(P)=2\mathcal{Q}'$ measures distance of the given correlation from the $\mathcal{Q}=0$ box, 
$P_{\mathcal{Q}=0}$, in the  canonical decomposition.  
\end{proof}
\subsection{$\mathcal{T}$ measure}
Upto local reversible operations any quantum correlation can be decomposed as a convex mixture of a PR-box and a Mermin-box and a restricted local box,
\be
P=\mathcal{G}'P^{000}_{PR}+\mathcal{Q}'\left(\frac{P^{000}_{PR}+P^{11\gamma}_{PR}}{2}\right)+(1-\mathcal{G}'-\mathcal{Q}')P^{\mathcal{G}=0}_{\mathcal{Q}=0}, \label{CQd}
\ee
where $\frac{1}{2}\left(P^{000}_{PR}+P^{11\gamma}_{PR}\right)$ are the two Mermin boxes canonical to the PR-box, $P^{000}_{PR}$,
and $P^{\mathcal{G}=0}_{\mathcal{Q}=0}$ is the local box which has $\mathcal{G}=\mathcal{Q}=0$. The local box in this decomposition is, in general, a nonproduct
box and therefore possesses classical correlations. Thus, the canonical decomposition implies that total nonclassical correlation 
in a given quantum joint probability distribution
is a sum of Bell discord and Mermin discord.

The observation that $\mathcal{G}$ and $\mathcal{Q}$ measure the distance of the given box from the corresponding 
$\mathcal{G}=0$ box and $\mathcal{Q}=0$ box invites us to define the quantity $\mathcal{T}$ that gives
distance of the given quantum box from the corresponding uncorrelated box which is a product of the marginals of the given box. 
\begin{definition}
$\mathcal{T}$ is defined as follows,
\be
\mathcal{T}=\max_{\alpha\beta}\mathcal{T}_{\alpha\beta}. \label{tot}
\ee
Here,
\be
\mathcal{T}_{\alpha\beta}=|\mathcal{B}_{\alpha\beta}-\mathcal{B}^{prod}_{\alpha\beta}|, \nonumber
\ee
where, 
\ba
\mathcal{B}^{prod}_{\alpha\beta}&=&|\braket{A_0}\braket{B_0}+(-1)^\beta\braket{A_0}\braket{B_1}\nonumber
\\&&+(-1)^\alpha\braket{A_1}\braket{B_0}+(-1)^{\alpha\oplus\beta\oplus1}\braket{A_1}\braket{B_1}|.\nonumber
\ea 
\end{definition}
This measure has the following properties:
\begin{enumerate}
\item $\mathcal{T}\ge0$.
 \item $\mathcal{T}=0$ iff the box is product i.e., $P(a_m,b_n|A_i,B_j)=P_A(a_m|A_i)P_B(b_n|B_j)$.
\begin{proof}
 Since $\mathcal{B}_{\alpha\beta}=\mathcal{B}^{prod}_{\alpha\beta}$ for the product box, $\mathcal{T}_{\alpha\beta}=0$ $\forall$ $\alpha\beta$.
For any box that can not written in the product form, $\mathcal{B}_{\alpha\beta}\ne\mathcal{B}^{prod}_{\alpha\beta}$ which, in turn, implies that $\mathcal{T}_{\alpha\beta}>0$
for any nonproduct box.
\end{proof}
\item Maximization in Eq. (\ref{tot}) makes $\mathcal{T}$ invariant under LRO and permutation of the parties.
As the canonical decomposition for quantum correlations in Eq. (\ref{CQd}) implies that $\max\mathcal{B}_{\alpha\beta}$ 
contains the total amount of nonclassicality in the given JPD, maximization is used in Eq. (\ref{tot}) rather than minimization.
\begin{proof}
Under local reversible operations and the permutation of the parties $\mathcal{T}_{\alpha\beta}$ in Eq. (\ref{tot}) transform into each
other.
\end{proof}
\end{enumerate}
\begin{thm}
As a consequence of the three properties of $\mathcal{T}$ given above, the additivity relation for quantum correlations follows,
\be
\mathcal{T}=\mathcal{G}+\mathcal{Q}\pm\mathcal{C}.
\ee
Here $\mathcal{C}$ quantifies classical correlations.
\end{thm}
\begin{proof}
Consider the correlation given by the canonical decomposition in Eq. (\ref{CQd}). Since this correlation maximizes $\mathcal{B}_{00}$,
\ba
\mathcal{T}(P)&=&|\mathcal{B}_{00}(P)-\mathcal{B}^{prod}_{00}(P)|\nonumber\\
\!&=&\!\left|4\mathcal{G}'\!+\!2\mathcal{Q}'\!+\!\left(1-\mathcal{G}'-\mathcal{Q}'\right)
\left(\mathcal{B}_{00}\left(P^{\mathcal{Q}=0}_{\mathcal{G}=0}\right)-\mathcal{B}^{prod}_{00}\left(P^{\mathcal{Q}=0}_{\mathcal{G}=0}\right)\right)\right|\nonumber\\
&=&\mathcal{G}+\mathcal{Q}\pm\mathcal{C},
\ea
where 
\be
\mathcal{C}=\left(1-\mathcal{G}'-\mathcal{Q}'\right)\left|\mathcal{B}_{00}\left(P^{\mathcal{Q}=0}_{\mathcal{G}=0}\right)-\mathcal{B}^{prod}_{00}\left(P^{\mathcal{Q}=0}_{\mathcal{G}=0}\right)\right|.
\ee
\end{proof}
\section{Quantum correlations}\label{QC}
Here we study total correlations in the quantum boxes obtained by spin projective measurements on the two-qubit systems: Alice performs measurements $A_i=\hat{a}_i\cdot\vec{\sigma}$
on her qubit along the two directions $\hat{a}_i$ and Bob performs measurements $B_j=\hat{b}_j\cdot\vec{\sigma}$
on her qubit along the two directions $\hat{b}_j$.  
Quantum theory predicts the behavior of the given box as follows,
\begin{equation}
P(a_m,b_n|A_i,B_j)=\tr\left(\rho_{AB}\Pi_{A_i}^{a_m}\otimes\Pi_{B_j}^{b_n}\right),
\end{equation}
where $\rho_{AB}$ is a density operator in the Hilbert space $\mathcal{H}^2_A\otimes \mathcal{H}^2_B$, $\Pi^{a_{m}}_{A_i}={1/2}\{\openone+a_{m}\hat{a}_i \cdot \vec{\sigma}\}$ and $\Pi^{b_{n}}_{B_j}={1/2}\{\openone+b_{n}\hat{b}_j \cdot \vec{\sigma}\}$, are the projectors generating binary outcomes $a_{m},b_{n} \in \{-1,1\}$.

Any quantum-quantum state which is neither a classical-quantum state nor a quantum-classical state gives rise to (1) a Bell discordant box which 
has $\mathcal{G}>0$ and $\mathcal{Q}=0$, (2) a Mermin discordant box
which has $\mathcal{G}=0$ and $\mathcal{Q}>0$, and (3) a Bell-Mermin discordant box which has $\mathcal{G}>0$ and $\mathcal{Q}>0$, for three different incompatible measurements \cite{Jeba}.
The set of $\mathcal{G}=\mathcal{Q}=0$ correlations forms a nonconvex subset of local correlations. The set of quantum correlations that
violate a Bell-CHSH inequality is a subset of $\mathcal{G}>0$ correlations.
The set of quantum correlations that violate an EPR-steering inequality \cite{CJWR},
\be
\mathcal{M}_{\alpha\beta\gamma}\le\sqrt{2},
\ee
where $[A_0, A_1]=-1$ or  $[B_0, B_1]=-1$, is a subset of $\mathcal{Q}>0$ correlations. The measurements that gives rise to maximal violation of a Bell-CHSH inequality 
(the Tsirelson bound) does not give rise to the violation of an EPR-steering inequality and vice versa due to the monogamy between nonlocality and contextuality,
\be
\mathcal{G}+2\mathcal{Q}\le4. \label{mGQ}
\ee
For general incompatible measurements, quantum correlations arising from the entangled states violate a Bell-CHSH inequality and an EPR-steering inequality simultaneously,
however, trade-off exists between the amount of nonlocality and the amount of contextuality as given by the above relation.   
This trade-off relation is analogous to the trade-off relation between KCBS and 
Bell-CHSH inequalities derived in Ref. \cite{mCNL} in the sense that both reveals monogamy between contextuality and nonlocality.

Since the correlations arising from the product states, $\rho_{AB}=\rho_A \otimes \rho_B$, factorize as the product of marginals corresponding to Alice and Bob,
they have $\mathcal{T}=0$. The set of $\mathcal{T}=0$ boxes is a subset of $\left\{P^{\mathcal{G}=0}_{\mathcal{Q}=0}\right\}$. Any nonproduct state can give rise to nonzero $\mathcal{T}$.
The set of $\mathcal{G}>0$ boxes and $\mathcal{Q}>0$ boxes are the subset of $\mathcal{T}>0$ boxes.
\subsection{Maximally entangled state}
Consider the correlations arising from the Bell state,
\be
\ket{\psi^+}=\frac{1}{\sqrt{2}}(\ket{00}+\ket{11}),
\ee
for the measurement settings: 
${\vec{a}_0}=\hat{x}$, ${\vec{a}_1}=\hat{y}$,
${\vec{b}_0} =\sqrt{p}\hat{x}-\sqrt{1-p}\hat{y}$ and ${\vec{b}_1}=\sqrt{1-p}\hat{x}+\sqrt{p}\hat{y}$, where $\frac{1}{2}\le p \le1$.
The correlations can be decomposed as follows,
\be
P=\mathcal{G}'P^{000}_{PR}+\mathcal{Q}'\left(\frac{P^{000}_{PR}+P^{110}_{PR}}{2}\right)+(1-\mathcal{G}'-\mathcal{Q}')P_N, \label{meb}
\ee 
where $\mathcal{G}'=\sqrt{1-p}$ 
and $\mathcal{Q}'=\sqrt{p}-\sqrt{1-p}$. These correlations violate the Bell-CHSH inequality i.e., $\mathcal{B}_{00}=2\left(\sqrt{p}+\sqrt{1-p}\right)>2$ if
$p\ne1$ and violate the EPR-steering inequality i.e., $\mathcal{M}_{11}=2\sqrt{p}>\sqrt{2}$ if $p\ne\frac{1}{2}$. 
Since the correlation maximally violates the Bell-CHSH inequality when $p=\frac{1}{2}$, each pair in the ensemble of two-qubits exhibits
nonlocality for the chosen measurements \cite{EPR2}. When $p$ is increased from $\frac{1}{2}$ to $1$, the number of pairs exhibiting nonlocality decreases and goes to zero 
when $p=1$.  However, the correlation maximally violates the EPR-steering inequality when $p=1$ which implies that each pair in the ensemble of two-qubits 
exhibits local contextuality as the measurements gives rise to bipartite version of the GHZ paradox \cite{UNLH,GHZ}. 
If $p$ is decreased from $1$ to $\frac{1}{2}$, the number of pairs exhibiting local contextuality decreases and the number of pairs 
exhibiting nonlocality increases as the violation EPR-steering inequality decreases and the violation of Bell-CHSH inequality increases.
The total amount of correlations in the JPD given in Eq. (\ref{meb}) is quantified by,
\be
\mathcal{T}=2\left(\sqrt{p}+\sqrt{1-p}\right)=\mathcal{G}+\mathcal{Q}=\left\{\begin{array}{lr}
\mathcal{G} \quad \text{when} \quad p=\frac{1}{2}\\ 
\mathcal{Q} \quad \text{when} \quad p=1\\ 
\end{array},
\right.
\ee
which implies that the JPD does not have the component of classically correlated box.
When the chosen measurements are performed on the ensemble of two-qubits,
each pair in a fraction of the ensemble quantified by $\mathcal{Q}'$ behaves contextually, each pair in a fraction of the ensemble quantified by $\sqrt{2}\mathcal{G}'$ 
behaves nonlocally and the remaining fraction behaves as noise.
\subsection{Schmidt states}
Consider the correlations arising from the Schmidt states (pure nonmaximally entangled states) \cite{Sch}:
\be
\rho_{S}\!=\!\frac{1}{4}\!\left(\!\openone\!\otimes\!\openone\!+\!c(\!\sigma_z\!\otimes\!\openone\!
+\!\openone\!\otimes\!\sigma_z\!)\!+\!s(\!\sigma_x\!\otimes\!\sigma_x\!-\!\sigma_y\!\otimes\!\sigma_y\!)
\!+\!\sigma_z\!\otimes\!\sigma_z\!\right)\!, \label{Schmidt}
\ee
where $c=\cos2\theta$, $s=\sin2\theta$ and $0\le\theta\le\frac{\pi}{4}$. 
\subsubsection{Bell-Schmidt box}
\textit{(i) Maximally mixed marginals correlations:-}
 The Schmidt states give to the noisy PR-box:
\ba
P=s\left[\frac{1}{\sqrt{2}} P^{000}_{PR}+\left(1-\frac{1}{\sqrt{2}}\right)P_N\right]+(1-s)P_N, \label{BSb}
\ea    
for the measurement settings: 
${\vec{a}_0}=\hat{x}$, ${\vec{a}_1}=\hat{y}$,
${\vec{b}_0} =\frac{1}{\sqrt{2}}(\hat{x}-\hat{y})$ and ${\vec{b}_1}=\frac{1}{\sqrt{2}}(\hat{x}+\hat{y})$.
These correlations violate the Bell-CHSH inequality i.e., $\mathcal{B}_{00}=2\sqrt{2}s>2$ if $s>\frac{1}{\sqrt{2}}$. 
Since the local box in Eq. (\ref{BSb}) gives $\mathcal{B}_{00}=0$,
violation of a Bell-CHSH inequality is not achieved by entanglement when $0<p\le\frac{1}{\sqrt{2}}$.
The correlations have,
\be
\mathcal{T}=\mathcal{G}=2\sqrt{2}s,
\ee
which implies that both $\mathcal{T}$ and $\mathcal{G}$ measure the distance of the box from white noise. For this measurement settings, a fraction of the ensemble quantified
by $s$ exhibits nonlocality purely and the remaining fraction behaves as white noise. 

\textit{(i) Nonmaximally mixed marginals correlations:-}
For the Popescu-Rohrlich measurement settings \cite{PRQB}: ${\vec{a}_0}=\hat{z}$, ${\vec{a}_1}=\hat{x}$,
${\vec{b}_0}=\cos t\hat{z}+\sin t\hat{x}$ and ${\vec{b}_1}=\cos t\hat{z}-\sin t\hat{x}$, 
where $\cos t=\frac{1}{\sqrt{1+s^2}}$, the correlations admit the following decomposition,
\ba
P&=&s^2\left[\frac{1}{\sqrt{1+s^2}}P_{PR}+\left(1-\frac{1}{\sqrt{1+s^2}}\right)P_N\right]\nonumber\\
&&+\left(1-s^2\right)P^{\mathcal{G}=0}_L(\rho).\label{PRQ}
\ea
Here $P^{\mathcal{G}=0}_L(\rho)$ is a nonmaximally mixed marginals box arising from the product state, 
\be
\rho=\rho_A \otimes \rho_B, \label{ScPr}
\ee
where
\be
\rho_A=\rho_B=\frac{1}{2}\left[1+\frac{c}{1-s^2}\right]\ket{0}\bra{0}+\frac{1}{2}\left[1-\frac{c}{1-s^2}\right]\ketbra{1}{1}\nonumber.
\ee
The $\mathcal{G}=0$ box in this decomposition is responsible for the violation of the Bell inequality when $0< s\le\frac{1}{\sqrt{2}}$;
as the box is already lifted to the face of the Bell polytope when $s=0$, any tiny amount of entanglement can give rise to the 
violation of the Bell-CHSH inequality i.e., $\mathcal{B}_{00}=2\sqrt{1+s^2}>2$ if $s>0$. 
The correlations have,
\be
\mathcal{T}=\mathcal{G}=\frac{4s^2}{\sqrt{1+s^2}}.
\ee
That is both $\mathcal{G}$ and $\mathcal{T}$ measure the distance of the box from the local box in the canonical decomposition as 
$P^{\mathcal{G}=0}_L$ in Eq. (\ref{PRQ}) is a product box.
Despite the correlations in Eq. (\ref{BSb}) do not violate the Bell-CHSH inequality when $0< s\le\frac{1}{\sqrt{2}}$, they have more nonlocality
than the correlatios in Eq. (\ref{PRQ}) as the former correlations have more irreducible PR-box component than the latter correlations. 
When the Popescu-Rohrlich measurements
are performed on the Schmidt state, a fraction of ensemble quantified by $\frac{\sqrt{2}s^2}{\sqrt{1+s^2}}$ exhibits nonlocality purely and the pairs in the 
remaining fraction are 
uncorrelated.

For the settings ${\vec{a}_0}=\hat{z}$, ${\vec{a}_1}=\hat{x}$,
${\vec{b}_0}=\frac{1}{\sqrt{2}}(\hat{z}+\hat{x})$ and ${\vec{b}_1}=\frac{1}{\sqrt{2}}(\hat{z}-\hat{x})$, the correlations can be decomposed as follows,
\be
P=s\left[\frac{1}{\sqrt{2}}P_{PR}+\left(1-\frac{1}{\sqrt{2}}\right)P_N\right]+(1-s)P^{\mathcal{G}=0}_L(\rho),\label{ZSb}
\ee
where $P^{\mathcal{G}=0}_L(\rho)$ arises from the correlated state $\rho=\frac{1}{2}\left(1+\frac{c}{1-s}\right)\ket{00}\bra{00}+\frac{1}{2}\left(1-\frac{c}{1-s}\right)\ket{11}\bra{11}$.
The difference between this box and the box in Eq. (\ref{BSb}) is that the local box in Eq. (\ref{ZSb}) is not a product box. 
The correlations violate the Bell inequality i.e., 
$\mathcal{B}_{00}=\sqrt{2}(1+s)>2$ if $s>\sqrt{2}-1$; since the local box in Eq. (\ref{ZSb}) is nonproduct, more entangled states violate the Bell inequality compared
to the correlations in Eq. (\ref{BSb}). 
The correlations have $\mathcal{G}=2\sqrt{2}s$ and
$\mathcal{T}=\sqrt{2}s(1+s)$. Since the JPD has the component of classical correlated box, it has $\mathcal{T}\ne\mathcal{G}$. 
The classical correlations is quantified by,
\be
\mathcal{C}=\mathcal{G}-\mathcal{T}=\sqrt{2}s(1-s)>0 \quad \text{when} \quad s\ne0,1.
\ee
Thus, a fraction of the ensemble given by $s$ exhibits nonlocality purely and the pairs in the remaining fraction exhibits classical correlations.  
\subsubsection{Mermin-Schmidt box}
(i) For the settings
${\vec{a}_0}=\hat{x}$, ${\vec{a}_1}=-\hat{y}$,
${\vec{b}_0}=\hat{y}$ and ${\vec{b}_1}=\hat{x}$, the Schmidt states give rise to the noisy Mermin-box:
\be
P=s \left(\frac{P^{000}_{PR}+P^{111}_{PR}}{2}\right)+(1-s) P_N, \label{MSb}
\ee 
which violates the EPR-steering inequality i.e., $\mathcal{M}_{00}=2s>\sqrt{2}$ if $s>\frac{1}{\sqrt{2}}$.
These local correlations have,
\be
\mathcal{T}=\mathcal{Q}=2s,
\ee
which implies that a fraction of the ensemble quantified by $s$ behaves contextually and the remaining fraction behaves as white noise.

(ii) For the settings
${\vec{a}_0}=\frac{1}{\sqrt{2}}(\hat{z}+\hat{x})$, ${\vec{a}_1}=\frac{1}{\sqrt{2}}(\hat{z}-\hat{x})$,
${\vec{b}_0}=\cos t\hat{z}-\sin t\hat{x}$, and ${\vec{b}_1}=\cos t\hat{z}+\sin t\hat{x}$, 
where $\cos t=\frac{1}{\sqrt{1+s^2} }$, the correlations can be decomposed as follows,
\ba
P&=&s^2\left[\frac{\sqrt{2}}{\sqrt{1+s^2}}\left(\frac{P^{000}_{PR}+P^{111}_{PR}}{2}\right)+\left(1-\frac{\sqrt{2}}{\sqrt{1+s^2}}\right)P_N\right]\nonumber\\
&&+\left(1-s^2\right)P_{\mathcal{Q}=0}(\rho), \label{CSB}
\ea
where $P_{\mathcal{Q}=0}(\rho)$ is a local box arising from the product state in Eq. (\ref{ScPr}). 
Since the local box, $P_{\mathcal{Q}=0}$, in this decomposition gives local bound when $s=0$, the box violates the EPR-steering inequality i.e., 
$\mathcal{M}_{00}=\sqrt{2}\sqrt{1+s^2}>\sqrt{2}$ if $s>0$. The box has,
\be
\mathcal{T}=\mathcal{Q}=\frac{2\sqrt{2}s^2}{\sqrt{1+s^2}}.
\ee
Since the $\mathcal{Q}=0$ box in Eq. (\ref{CSB}) is a product box, the amount of total correlations equals to Mermin discord. Notice that the correlations in Eq .(\ref{MSb}) 
have more Mermin discord than that for the correlations 
in Eq. {\ref{CSB}} which implies that the latter correlations have less amount contextuality than the former correlations.   

For the settings
${\vec{a}_0}=\frac{1}{\sqrt{2}}(\hat{z}+\hat{x})$, ${\vec{a}_1}=\frac{1}{\sqrt{2}}(\hat{z}-\hat{x})$,
${\vec{b}_0}=\frac{1}{\sqrt{2}}(\hat{z}-\hat{x})$, and ${\vec{b}_1}=\frac{1}{\sqrt{2}}(\hat{z}+\hat{x})$, 
the Schmidt states give rise to the following correlation,
\be
P=s\left(\frac{P^{000}_{PR}+P^{111}_{PR}}{2}\right)+(1-s)P^{\mathcal{G}=0}_L(\rho),\label{CSB1}
\ee
where $P^{\mathcal{G}=0}_L(\rho)$ arises from the correlated state $\rho=\frac{1}{2}\left(1+\frac{c}{1-s}\right)\ket{00}\bra{00}+\frac{1}{2}\left(1-\frac{c}{1-s}\right)\ket{11}\bra{11}$.
This box violates the EPR-steering inequality i.e., $\mathcal{M}_{00}=(1+s)>\sqrt{2}$ if $s>\sqrt{2}-1$ which is larger violation than that for the box in Eq. (\ref{MSb}). 
The correlations have $\mathcal{T}=s(1+s)$ and $\mathcal{Q}=2s$ which implies
that the classical correlations in the JPD is quantified as follows,
\be
\mathcal{C}=\mathcal{Q}-\mathcal{T}=s(1-s)>0 \quad \text{when} \quad s\ne0, 1.
\ee
\subsubsection{Bell-Mermin-Schmidt box}
(i) The correlations can be decomposed as follows:
\ba
P&=&\left(1-q-g\right) P_N+ \frac{q}{2}\left(P^{000}_{PR}+P^{11\gamma}_{PR}\right)\nonumber\\
&+&g\left[\frac{1}{\sqrt{2}}P^{000}_{PR}+\left(1-\frac{1}{\sqrt{2}}\right)P_N\right], \label{BMSb}
\ea
for the settings:
${\vec{a}_0}=s\hat{x}+c\hat{y}$,
${\vec{a}_1}=c\hat{x}-s\hat{y}$,
${\vec{b}_0}=\frac{1}{\sqrt{2}}(\hat{x}+\hat{y})$ and
${\vec{b}_1}=\frac{1}{\sqrt{2}}(\hat{x}-\hat{y})$, where
$q=\frac{s\left||c+s|-|c-s|\right|}{\sqrt{2}}$ and $g=|s(s-c)|$.
This box gives,
\be
\mathcal{G}=2\sqrt{2}s|s-c|>0 \quad \text{except when} \quad \theta  \ne0, \frac{\pi}{8}, \nonumber
\ee
\ba
\mathcal{Q}&=&s\sqrt{2}\Big||c+s|-|c-s|\Big|>0 \quad \text{except when} \quad \theta \ne0, \frac{\pi}{4}\nonumber\\
&=&\left\{\begin{array}{lr}
2\sqrt{2}s^2 \quad \text{when} \quad c>s\\ 
2\sqrt{2}cs \quad \text{when} \quad s>c\\ 
\end{array}
\right.\nonumber
\ea
and
\ba
\mathcal{T}&=&\left\{
\begin{array}{lr}
2\sqrt{2}s^2 \quad \text{when} \quad s>c\\ 
2\sqrt{2}cs \quad \text{when} \quad c>s\\ 
\end{array}
\right.\nonumber\\ 
&=&\mathcal{G}+\mathcal{Q},
\ea
which implies that the box has nonclassical correlations purely as the box does not have classical correlation component; 
a fraction of the ensemble quantified by $g$ exhibits nonlocality purely, a fraction of the ensemble quantified by $q$ exhibits contextuality and 
the remaining fraction behaves as white noise. 

(ii) For the settings: ${\vec{a}_0}=c\hat{x}+s\hat{z}$,
${\vec{a}_1}=s\hat{x}-c\hat{z}$,
${\vec{b}_0}=\frac{1}{\sqrt{2}}(\hat{x}+\hat{z})$ and
${\vec{b}_1}=\frac{1}{\sqrt{2}}(-\hat{x}+\hat{z})$, the correlations have the same amount of $\mathcal{G}$ and $\mathcal{Q}$ as for the correlations in Eq. (\ref{BMSb}), 
however, they have different amount of $\mathcal{T}$ which is given as follows,
\ba
\mathcal{T}&=&\left\{
\begin{array}{lr}
\sqrt{2}s^2(1+s) \quad \text{when} \quad s>c\\ 
\sqrt{2}cs(1+s) \quad \text{when} \quad c>s.\\ 
\end{array}
\right.\nonumber
\ea
Thus, the correlations arising from the latter settings (ii) have the decomposition which has the same amount of PR-box and Mermin box components 
as for the former settings (i) except that white noise in Eq. (\ref{BMSb}) 
is replaced by the classically correlated box.
The classical correlations is quantified by, 
\ba
\mathcal{C}&=&\mathcal{G}+\mathcal{Q}-\mathcal{T}\nonumber\\
&=&\left\{
\begin{array}{lr}
\sqrt{2}s^2(1-s) \quad \text{when} \quad s>c\\ 
\sqrt{2}cs(1-s) \quad \text{when} \quad c>s.\\ 
\end{array}
\right.\nonumber
\ea
\subsection{Werner states} 
Consider the correlations arising from the Werner states \cite{Werner},
\be
\rho_W=p\ketbra{\psi^+}{\psi^+}+(1-p)\frac{\openone}{4}.
\ee
The Werner states are entangled if $p>\frac{1}{3}$ and have nonzero quantum discord if $p>0$ \cite{WZQD}. 
Since the Werner states have component of irreducible entangled state  if $p>0$, they can give rise to nonclassical correlations if $p>0$.
As the Werner states can only give rise to maximally mixed marginals correlations,
the nonclassical correlations arising from the Werner states can not have component of classical correlation. 
\subsubsection{Bell-Werner box}
The correlations have the following decomposition,
\be
P=p\left[\frac{1}{\sqrt{2}} P^{000}_{PR}+\left(1-\frac{1}{\sqrt{2}}\right)P_N\right]+(1-p)P_N.
\ee
for the settings that corresponds to the correlation in Eq. (\ref{BSb}). These correlations have,
\be
\mathcal{T}=\mathcal{G}=2\sqrt{2}p.
\ee
\subsubsection{Mermin-Werner box}
The Werner states give rise to the noisy Mermin box,
\be
P=(1-p) P_N+ p\left(\frac{P^{000}_{PR}+P^{111}_{PR}}{2}\right),
\ee 
for the settings corresponding to the correlation in Eq. (\ref{MSb}).
These correlations have,
\be
\mathcal{T}=\mathcal{Q}=2p.
\ee
\subsubsection{Bell-Mermin-Werner box}
Th correlations admit the following decomposition:
\be
P=(1-q-r) P_N+ \frac{q}{2}\left(P^{000}_{PR}+P^{11\gamma}_{PR}\right)+|r|P^{000}_{PR}, \label{BMWb}
\ee
for the settings:
${\vec{a}_0}=\sqrt{p}\hat{x}+\sqrt{1-p}\hat{y}$,
${\vec{a}_1}=\sqrt{1-p}\hat{x}-\sqrt{p}\hat{y}$,
${\vec{b}_0}=\frac{1}{\sqrt{2}}(\hat{x}+\hat{y})$ and
${\vec{b}_1}=\frac{1}{\sqrt{2}}(\hat{x}-\hat{y})$, where
$q=p\sqrt{2(1-p)}$ and $r=\frac{1}{\sqrt{2}}p\left(\sqrt{p}-\sqrt{1-p}\right)$. The box gives
\ba
\mathcal{G}&=&2\sqrt{2}p|\sqrt{p}-\sqrt{1-p}|>0 \quad \text{except when} \quad p\ne0, \frac{1}{2} \nonumber,\\
\mathcal{Q}&=&\sqrt{2}p\Big|\sqrt{p}+\sqrt{1-p}-\left|\sqrt{p}-\sqrt{1-p}\right|\Big|\nonumber\\
&&>0 \quad \text{except when} \quad p\ne0, 1\nonumber\\
&=&\left\{\begin{array}{lr}
2p\sqrt{2p}\quad \text{when} \quad 0\le p\le\frac{1}{2}\nonumber\\ 
2p\sqrt{2(1-p)} \quad \text{when}  \quad\frac{1}{2}\le p\le1\nonumber
\end{array}
\right.
\ea
and
\ba
\mathcal{T}&=&\left\{\begin{array}{lr}
2p\sqrt{2(1-p)}\quad \text{when} \quad 0\le p\le\frac{1}{2}\nonumber\\ 
2p\sqrt{2p}\quad \text{when}  \quad\frac{1}{2}\le p\le1\nonumber
\end{array}
\right.\\
&=&\mathcal{G}+\mathcal{Q}.
\ea
\subsection{Mixture of maximally entangled state with colored noise}
Consider the correlations arising from the mixture of the Bell state and the classically correlated state,
\be
\rho= p \ketbra{\psi^+}{\psi^+}+(1-p) \rho_{CC}, \label{rCC}
\ee
where $\rho_{CC}=\frac{1}{2}(\ketbra{00}{00}+\ketbra{11}{11})$. 
Only when suitable incompatible measurements in the \plane{x}{z} are performed on these states, they have different nonclassical behavior than the Werner states. 
\subsubsection{Bell discordant box}
For the settings that gives rise to the noisy PR-box in Eq. (\ref{BSb}), 
\be
\mathcal{T}=\mathcal{G}=2\sqrt{2}p.
\ee

The measurement settings, 
${\vec{a}_0}=\hat{z}$, ${\vec{a}_1}=\hat{x}$,
${\vec{b}_0}=\cos t\hat{z}+\sin t\hat{x}$ and ${\vec{b}_1}=\cos t\hat{z}-\sin t\hat{x}$, where $\cos t=\frac{1}{\sqrt{1+p^2}}$, gives rise to the 
violation of the Bell inequality, $\mathcal{B}_{00}=2\sqrt{1+p^2}>2$, if $p>0$. Since the box lies at the face of the Bell polytope when $p=0$, any tiny amount of 
entanglement gives rise to the violation Bell-CHSH inequality.
The correlations have $\mathcal{G}=\frac{4p^2}{\sqrt{1+p^2}}$ and $\mathcal{T}=2\sqrt{1+p^2}$ which implies that the classical correlations is quantified as follows,
\be
\mathcal{C}=\mathcal{T}-\mathcal{G}=\frac{2(1-p^2)}{\sqrt{1+p^2}}.
\ee
\subsubsection{Mermin discordant box}
The measurement settings, 
${\vec{a}_0}=\frac{1}{\sqrt{2}}(\hat{z}+\hat{x})$, ${\vec{a}_1}=\frac{1}{\sqrt{2}}(\hat{z}-\hat{x})$,
${\vec{b}_0}=\cos t\hat{z}+\sin t\hat{x}$ and ${\vec{b}_1}=\cos t\hat{z}-\sin t\hat{x}$, where $\cos t=\frac{1}{\sqrt{1+p^2}}$, gives rise to the
violation of the EPR-steering inequality, $\mathcal{M}_{00}=\sqrt{2}\sqrt{1+p^2}>\sqrt{2}$, if $p>0$. The correlations have 
$\mathcal{Q}=\frac{2\sqrt{2}p^2}{\sqrt{1+p^2}}$ and $\mathcal{T}=\sqrt{2}\sqrt{1+p^2}$ which implies that the amount of classical correlations in the JPD is 
quantified as follows,
\be
\mathcal{C}=\mathcal{T}-\mathcal{Q}=\frac{\sqrt{2}(1-p^2)}{\sqrt{1+p^2}}.
\ee
\subsubsection{Bell-Mermin discordant box}
For the measurement settings
${\vec{a}_0}=\sqrt{p}\hat{z}+\sqrt{1-p}\hat{x}$,
${\vec{a}_1}=\sqrt{1-p}\hat{z}-\sqrt{p}\hat{x}$,
${\vec{b}_0}=\frac{1}{\sqrt{2}}(\hat{z}+\hat{x})$ and
${\vec{b}_1}=\frac{1}{\sqrt{2}}(\hat{z}-\hat{x})$,
the correlations have the same amount of Bell discord and Mermin discord as for the correlations in Eq. (\ref{BMWb}), however, the box has different amount of 
total correlations,
\ba
\mathcal{T}&=&\left\{\begin{array}{lr}
(1+p)\sqrt{2(1-p)}\quad \text{when} \quad 0\le p\le\frac{1}{2}\nonumber\\ 
(1+p)\sqrt{2p}\quad \text{when}  \quad\frac{1}{2}\le p\le1\nonumber
\end{array}
\right.\\
&>&\mathcal{G}+\mathcal{Q},
\ea
because of the classically correlated noise. The amount of classical correlations is given by,
\ba
\mathcal{C}&=&\mathcal{T}-\mathcal{G}-\mathcal{Q}\nonumber\\
&=&\left\{
\begin{array}{lr}
(1-p)\sqrt{2(1-p)}\quad \text{when} \quad 0\le p\le\frac{1}{2}\nonumber\\ 
(1-p)\sqrt{2p}\quad \text{when}  \quad\frac{1}{2}\le p\le1.\nonumber
\end{array}
\right.\nonumber
\ea
\section{Conclusions}\label{disconc}
We interpreted Bell discord and Mermin discord as distance measures for nonlocality and contextuality which led us to construct the distance measure, $\mathcal{T}$, 
which is zero iff the box is product. 
We have discussed the problem of separating the total correlations in the quantum boxes into nonlocality, contextuality and classical correlations using these
three measures. We have studied the additivity relation for quantum correlations in two-qubit systems.  
The distance measure interpretation has allowed us to understand why some entangled states cannot lead to the violation of a Bell-CHSH inequality.
\section*{Acknowledgements}
I thank IISER Mohali for financial support. I had useful discussions with Dr. Pranaw Rungta, and he suggested the problem.    

\end{document}